\documentclass{ingosclass}
\bibliography{bib}
\let\oldcdot\cdot %if we still need it 
\renewcommand{\cdot}{\mathbin{\mkern-2mu\oldcdot\mkern-2mu}}

\newcommand{\wrt}{w.\,r.\,t.}
\newcommand{\eg}{e.\,g.}
\newcommand{\soth}{s.\,t.} 
\newcommand{\ie}{i.\,e.}

\newcommand{\resp}{resp.}

\newcommand{\formComma}{\,\text{,}}
\newcommand{\formPeriod}{\,\text{.}}

\newcommand{\R}{\mathbb{R}}

\newcommand{\surf}{\mathcal{S}}

\newcommand{\tangent}[2][]{\tensor{\operatorname{T}\!}{#1}#2}
\newcommand{\tangentS}[1][]{\tangent[#1]{\surf}}
\newcommand{\tangentR}[1][]{\tangent[#1]{\R^3\vert_{\surf}}}

\newcommand{\tangentAS}[1][^2]{\tensor{\operatorname{\mathcal{A}}\!}{#1}\surf}

\newcommand{\tangentScal}{\tangentS[^0]}

\newcommand{\dbdot}{\operatorname{:}}

\newcommand{\timeLll}{\mathfrak{L}^{\flat\flat}}

\newcommand{\Dt}[1][]{\operatorname{D}^{#1}_t \!}

\newcommand{\Dlow}{\Dt[\flat]}

\newcommand{\Comp}{\mathsf{C}}
\newcommand{\nablaC}{\nabla_{\!\Comp}}

\newcommand{\Tr}{\operatorname{Tr}}
\renewcommand{\div}{\operatorname{div}}
\newcommand{\Div}{\operatorname{Div}}

\newcommand{\DivC}{\Div_{\!\Comp}}

\newcommand{\Grad}{\operatorname{Grad}}
\newcommand{\GradC}{\Grad_{\Comp}}

\newcommand{\hil}{\operatorname{L}^{\!2}}
\newcommand{\hilspace}[1]{\hil(#1)}
\newcommand{\inner}[3][0]{\left\langle #3 \right\rangle_{\hspace{-#1pt}#2}}
\newcommand{\normsq}[3][0]{\left\| #3 \right\|_{\hspace{-#1pt}#2}^2}
\newcommand{\innerH}[3][0]{\inner[#1]{\hilspace{#2}}{#3}}
\newcommand{\normHsq}[3][0]{\normsq[#1]{\hilspace{#2}}{#3}}

\newcommand{\paraC}{X}
\newcommand{\para}{\boldsymbol{\paraC}}
\newcommand{\normalC}{\nu}
\newcommand{\normal}{\boldsymbol{\normalC}}
\newcommand{\shopC}{I\!I}
\newcommand{\shop}{\boldsymbol{\shopC}}

\newcommand{\meanc}{\mathcal{H}}

\newcommand{\Id}{\boldsymbol{Id}}
\newcommand{\IdS}{\Id_{\surf}}
\newcommand{\IdTwo}{\Ib_{2}}

\newcommand{\mfrak}{\mathfrak{m}}
\newcommand{\ofrak}{\mathfrak{o}}

\newcommand{\eb}{\boldsymbol{e}}

\newcommand{\gb}{\boldsymbol{g}}

\newcommand{\nb}{\boldsymbol{n}}

\newcommand{\qb}{\boldsymbol{q}}
\newcommand{\rb}{\boldsymbol{r}}

\newcommand{\vb}{\boldsymbol{v}}
\newcommand{\wb}{\boldsymbol{w}}

\newcommand{\Eb}{\boldsymbol{E}}

\newcommand{\Gb}{\boldsymbol{G}}

\newcommand{\Ib}{\boldsymbol{I}}

\newcommand{\Rb}{\boldsymbol{R}}

\newcommand{\Vb}{\boldsymbol{V}}
\newcommand{\Wb}{\boldsymbol{W}}

\newcommand{\etab}{\boldsymbol{\eta}}
\newcommand{\thetab}{\boldsymbol{\theta}}      % oder \varthetab

\newcommand{\lambdab}{\boldsymbol{\lambda}}
\newcommand{\mub}{\boldsymbol{\mu}}

\newcommand{\Gammab}{\boldsymbol{\Gamma}}

\newcommand{\hfrak}{\mathfrak{h}}

\newcommand{\vnor}{v_{\bot}}
\newcommand{\wnor}{w_{\bot}}

\newcommand{\nullb}{\boldsymbol{0}}

\newcommand{\Vbobs}{\Vb_{\!\!\ofrak}}
\newcommand{\vbobs}{\vb_{\!\ofrak}}

\newcommand{\vF}{\vb^{\flat}}

\newcommand{\wF}{\wb^{\flat}}
\newcommand{\wS}{\wb^{\sharp}}

\newcommand{\potenergy}{\mathfrak{U}}

\renewcommand{\dh}{\partial h}
\newcommand{\ddh}{\partial^2 h}

\newcommand{\HT}{\textup{S}}
\newcommand{\energyHT}{\potenergy_{\HT}}

\newcommand{\ScalState}{\textup{S}}
\newcommand{\energyS}{\potenergy_{\ScalState}}

\newcommand{\deltafrac}[2]{\frac{\delta #1}{\delta #2}}
\newcommand{\ddfrac}[2][]{\frac{\textup{d} #1}{\textup{d} #2}}
\newcommand{\ddt}[1][]{\ddfrac[#1]{t}}

\newcommand{\dS}{\textup{d}\surf}

\newcommand{\DX}{D_{\!\para}}
\newcommand{\Dpsi}{D_{\!\psi}}

\newcommand{\Mspat}{M_{\para}}
\newcommand{\Mdens}{M_{\psi}}

\newcommand{\HmOne}{H^{\mathord{\text{-}}1}}

\title{Scalar Truesdell Time Derivative and \mbox{($L^2,\HmOne$) - Surface Gradient Flows}}
\author[1]{Rainer Backofen}
\author[1]{Ingo Nitschke}
\author[1,2,3]{Axel Voigt}
\affil[1]{Institut f{\"u}r Wissenschaftliches Rechnen, Technische Universit{\"a}t Dresden, 01062 Dresden, Germany}
\affil[2]{Center for Systems Biology Dresden (CSBD), Pfotenhauerstr. 108, 01307 Dresden, Germany}
\affil[3]{Cluster of Excellence Physics of Life (PoL), Technische Universität Dresden, 01062~Dresden, Germany}
\date{}

\begin{document}
\maketitle
%\tableofcontents

\begin{abstract}
We address surface gradient flows which allow for energy dissipation by evolving the surface and a scalar quantity on it, simultaneously. A proper choice of the time derivative and the gauge of surface independence guarantees energy dissipation and ensures conservation of the scalar quantity. The resulting system of partial differential equations couples geometric evolution equations for the evolution of the surface in normal directions, equations for tangential movement and scalar-valued equations on the evolving surface. We discuss the general setting and the special case of surface tension flows and numerically demonstrate the importance of tangential movement on the evolution.
\end{abstract}

\section{Introduction}

We consider gradient flows of surface energies $ \potenergy = \potenergy[\para,\psi] $ that depend on a surface $\surf$ through a parameterization $\para$ and a scalar field $\psi$ defined on $\surf$ and allow for dissipation by evolving $\para$ and $\psi$, simultaneously. Such models arise in various applications, ranging from material science to biology and medicine. Examples include adatoms on material surfaces \cite{FRIED20041, Burger_CMS_2006, Raetz_Nonl_2007,Caroccia_ARMA_2018}, particle densities in thin crystalline structures \cite{Aland_PF_2011,Aland_PRE_2012,Aland_MMS_2012,elder2021modeling,BenoitMarechalNitschkeEtAl_MoM_2024}, signaling networks and phase separation in biomembranes \cite{Wang_JMB_2008,Lowengrub_PRE_2009,Elliott_JCP_2010,Jilkine_PLOSCB_2011,Marth_JMB_2014} or tumor growth \cite{Chaplain_2001,Crampin_2002,landsberg2010multigrid,Barreira_2011,Eyles_2019}. In contrast to its broad applicability, mathematical foundations for such models are still limited. One issue arises from a mutual dependency of $\para$ and $\psi$. In \cite{NitschkeSadikVoigt_IJoAM_2023} a more general situation has been considered with the scalar field replaced by a tangential n-tensor field. In this situation the dependency of the tangential n-tensor field on the parameterization $\para$ is obvious and it has been shown that in order to guarantee energy dissipation a notion of surface independence has to be chosen consistently with the time derivative. This allows for various choices, which increase with $n$. While the equilibrium states are independent of these choices, the dynamics differs. As soon as the dynamics matters, a proper choice is required, which is typically determined by the specific application \cite{NitschkeSadikVoigt_IJoAM_2023,stone2023note,pollard2025gauge}. For scalar fields ($n=0$) the situation is much simpler. For ($L^2,L^2$) - surface gradient flows the natural choice is the material time derivative $\dot{\psi}$, see, \eg,  \cite{CERMELLI_FRIED_GURTIN_2005,Dziuk_Elliott_2013,NitschkeVoigt_JoGaP_2023} for proper definitions on evolving surfaces, and the material gauge of surface independence $\eth_{\Wb}\psi=0$, introduced in \cite[Def. 6]{NitschkeSadikVoigt_IJoAM_2023}. The situation changes for ($L^2, \HmOne$) - surface gradient flows, e.g., a $L^2$ - gradient flow with respect to the evolution of $\para$ and a $\HmOne$ - surface gradient flow with respect to the evolution of $\psi$. Besides energy dissipation by evolving $\para$ and $\psi$ simultaneously, such flows also should ensure conservation of $\psi$, e.g., $\frac{d}{dt} \int_\surf \psi d\surf = 0$. At least if local expansion or contraction of the surface are allowed the natural choice of the material time derivative and the material gauge of surface independence reaches its limits and a mutual dependency of $\para$ and $\psi$ become apparent also for scalar fields. We will demonstrate that considering both does not ensure conservation of $\psi$ and modifying the evolution to ensure the conservation property but considering the material gauge of surface independence does not guarantee energy dissipation. We introduce the scalar Truesdell time derivative $\mathring{\psi}$ and the Truesdell gauge of surface independence $\eth_{\Wb}\psi = -\psi\DivC\Wb$, see \cref{sec:STTD,sec:SGF} for definitions, to deal with this situation. At first glance it might look like a formal redefinition only allowing for a compact formulation, but it turns out to become a useful tool to construct ($L^2, \HmOne$) - surface gradient flows which guarantee energy dissipation and conservation property. We discuss the general setting and address surface tension flows as a particular example. The resulting system of equations couples geometric evolution equations for the evolution of $\surf$ in normal directions, equations for tangential movement and scalar-valued surface partial differential equations on the evolving surface. Especially the tangential flow, induced by gradients in $\psi$ requires attention, as it modifies the dynamics of the considered surface gradient flows and is not considered in most of the mentioned applications above. Analytical results for such systems are rare and mostly exist only for special situations, \eg, surface partial differential equations on (given) evolving surfaces, \eg, \cite{alphonse2015abstract}, geometric evolution equations not depending on $\psi$, \eg,  \cite{huisken1984flow}, and coupled systems of geometric evolution equations and surface partial differential equations without the tangential movement, \eg, \cite{,Caroccia_ARMA_2018,ABELS202314,ABELS2023236}. Also results in numerical analysis are restricted to special situations, \eg, the evolution of curves coupled with diffusion on it \cite{barrett2017numerical,pozzi2017curve}, a model for forced mean curvature flow coupled with diffusion on the surface \cite{kovacs2020convergent} and the coupled system resulting as a ($L^2,\HmOne$) - surface gradient flow, but without tangential movement \cite{elliott2022numerical}. We here refrain from extending these results and solely concentrate on modeling issues and resulting properties. We demonstrate the shown properties numerically using a height observer formulation. 

The paper is structured as follows: In \cref{sec:STTD} we introduce the scalar Truesdell time derivative, in \cref{sec:SGF} we consider ($L^2,\HmOne$) - surface gradient flows, including the general setting, surface tension flows as well as their height observer formulation. In \cref{sec:C} we discuss the numerical results and draw conclusions. 
Within the Appendix we consider a general approach to the height observer formulation without a small displacement assumption and address the corresponding models if tangential movement is neglected. Most of the notation is introduced on the fly as it becomes needed. A comprehensive overview of the used notation can be found in \cite{Nitschke_2025}, which consolidates the conventions from \cite{NitschkeVoigt_JoGaP_2023,NitschkeSadikVoigt_IJoAM_2023,NitschkeVoigt_AiDE_2025,NitschkeVoigt_PotRSAMPaES_2025}.

\section{Scalar Truesdell Time Derivative} \label{sec:STTD}

The space of time-dependent scalar fields on two-dimensional evolving or moving surfaces $ \surf\subset\R^3 $ is given by 
$ \tangentScal:=\{\psi:\mathcal{T}\times\surf\rightarrow \R\} $, 
where $ \mathcal{T}\subseteq\R $ is a time interval and we do not specify how the surface is accessed. 
Both local and global coordinates may be used for evaluation. Coordinate invariance provides considerable freedom in this regard.
We further consider $ \tangentR[^n]:=\{\Wb:\mathcal{T}\times\surf\rightarrow (\R^3)^n \} $ as the space of $ n $-tensor fields on $ \surf $ and $ \tangentS[^n] < \tangentR[^n] $ as the space of tangential $ n $-tensor fields on $ \surf $,
where ``$ < $'' denotes the subtensor field relation. 
For $\tangentS[^1]$ and $\tangentR[^1]$ we also write $\tangentS$ and $\tangentR$, respectively.

Conservation of a scalar field $\psi \in \tangentScal$ can be achieved by compatibility \wrt\ the transport theorem, see, \eg, \cite{Stone1990111,FRIED20041,Dziuk_Elliott_2013}. 
This states
\begin{align}\label{eq:transport_theorem}
        \ddt \int_\surf \psi \dS 
            &= \int_\surf \dot{\psi} + \psi\DivC\Vb \dS
            = \int_\surf \dot{\psi} + \psi\left( \div\vb - \vnor\meanc \right) \dS\formComma
\end{align}
with $\dot{\psi} = \partial_t \psi + (\Vb - \Vbobs)\nablaC\psi$ the material time derivative, see, \eg,  \cite{CERMELLI_FRIED_GURTIN_2005,Dziuk_Elliott_2013,NitschkeVoigt_JoGaP_2023} for proper definitions on evolving surfaces and \cite[Appendix B.3.]{BachiniKrauseEtAl_JoFM_2023} for their relation. In the above equation, $ \Vb=\vb+\vnor\normal\in\tangentR $ is the material velocity, with  $ \vb\in\tangentS $ the tangential material velocity, $ \vnor\in\tangentScal $ the normal velocity, and $ \normal\in\tangentR $ the (surface) normal field. Note that the normal velocity is equal for all observer of the same moving surface, thus it does not need to bear the name component ``material''. $ \Vbobs=\vbobs+\vnor\normal\in\tangentR $ is the observer velocity, which is again decomposed into $ \vbobs\in\tangentS $ the tangential observer velocity, and $ \vnor\in\tangentScal $ the normal velocity.
Note that the normal velocity is equal for all observer of the same moving surface, thus it does not need to bear the name component ``observer''. The observer velocity is arbitrary and can serve as mesh velocity or derived \wrt\ a height formulation. As a consequence $  \Vb - \Vbobs = \vb - \vbobs \in\tangentS $. The considered differential operators are  $ \div:\tangentS[^n]\rightarrow\tangentS[^{n-1}] $, which is the covariant divergence and could be defined by $ \div\rb = (\nabla\rb)\dbdot\IdS $, with $ \nabla:\tangentS[^n]\rightarrow\tangentS[^{n+1}] $ the covariant derivative, $ \IdS=\Id-\normal\otimes\normal\in\tangentS[^2] $ the surface identity and $ \rb\in\tangentS[^n] $. The other divergence $ \DivC: \tangentR[^n] \rightarrow \tangentR[^{n-1}] $ is the componentwise trace-divergence and could be defined by $ \DivC\Rb = \nablaC\Rb\dbdot\Id = \nablaC\Rb\dbdot\IdS $ for all $ \Rb\in\tangentR[^n] $ with $ \nablaC: \tangentR[^n] \rightarrow \tangentR[^n]\otimes\tangentS $ the componentwise (surface) derivative defined \wrt\ a Cartesian base $ \{\eb_{A}\} $ \soth\ $ \left[ \nablaC\Rb \right]^{A_1 \ldots A_n B} \eb_B = \nabla R^{A_1 \ldots A_n} $ for all $ \Rb\in\tangentR[^n] $, or $ \nablaC\Rb = \left( \nabla_{\R^3}\Rb \right)\vert_{\surf}\IdS $ for an arbitrary sufficiently smooth extension of $ \Rb $ in a vicinity of $ \surf $.
It holds $ \DivC\Wb = \div\wb - \wnor\meanc $ for all $ \Wb=\wb + \wnor\normal\in\tangentR $. Note that this divergence is not the $ \hil $-adjoint of $ \nablaC $. We therefore also use $ \GradC: \tangentR[^n] \rightarrow \tangentR[^{n+1}] $ as the componentwise adjoint gradient defined by the $ \hil $-adjoint of the trace-divergence, \ie\ $ \GradC := -\DivC^{*} $. It holds $ \GradC f = \DivC(f\IdS) = \nabla f + f\meanc \normal $ for all scalar fields $ f\in\tangentScal $. Finally $ \meanc = -\DivC\normal \in\tangentScal $ is the mean curvature. With these things defined one could be tempted to define a new time derivative, which we call the \emph{scalar Truesdell time derivative}, which acting on $ \psi\in\tangentScal $ reads
\begin{empheq}[box=\fbox]{align}\label{eq:truesdell_rate}
\begin{aligned}
    \mathring{\psi}
        &= \dot{\psi} + \psi\DivC\Vb
        &&= \partial_t\psi + (\Vb - \Vbobs)\nablaC\psi + \psi\DivC\Vb &&= \partial_t\psi + \DivC(\psi\Vb) - \Vbobs\nablaC\psi\\
        &= \dot{\psi} + \psi\left( \div\vb - \vnor\meanc \right)
        &&= \partial_t\psi + \nabla_{\vb-\vbobs}\psi + \psi\left( \div\vb - \vnor\meanc \right) &&= \partial_t\psi + \div(\psi\vb) - \nabla_{\vbobs}\psi - \psi\vnor\meanc\formPeriod
\end{aligned}
\end{empheq}
It has already been introduced in \cite{Federico_ZfaMuP_2022} for 3-dimensional Euclidean spaces $ \cong\R^3 $ as Truesdell rate. 
The considered covariant directional derivative reads $ \nabla_{\wb}\psi = \wb\nabla\psi = (\nabla\psi)\wb\in\tangentScal $ for all directions $ \wb\in\tangentS $. 
With the scalar Truesdell time derivative the conservation property simply reads
\begin{align*}
        \ddt \int_\surf \psi \dS 
            &= \int_\surf \mathring{\psi} \dS = 0\formPeriod
\end{align*}

While this is just a formal redefinition allowing for a compact formulation, there are other approaches to the Truesdell rate, which point to its deeper mathematical origin. When $\psi$ is interpreted not merely as a scalar field but as a density proxy, it is natural to introduce the associated differential 2-form $\thetab = \psi \mub \in \Lambda^2\surf$,
where $ \mub\in \Lambda^2\surf $ is the differential area form on $ \surf $ \wrt\ chosen local coordinates.
As a consequence, it holds $ \int_{\surf}\psi\dS = \int_{U}\thetab $, where $ U\subset\R^2 $ is the domain of the local coordinates.
Although $\psi$ and $\thetab$ are trivially isomorphic, their temporal change on moving surfaces differs fundamentally. 
In particular, scalar field rates are typically defined via the material (\resp\ total or substantial) time derivative, 
whereas differential forms are evolved using the lower-convected time derivative, given by the Lie derivative of time-dependent covariant tensor field proxies \cite{marsden1994mathematical}.
Since differential forms and skew-symmetric tensors are equivalent as multilinear maps into $ \R $, the lower-convected time derivative from \cite{NitschkeVoigt_JoGaP_2022,NitschkeVoigt_JoGaP_2023} applies directly in a coordinate-free formulation.
We can identify $ \mub $ by the Levi-Civita tensor $ \Eb\in\tangentAS $, where $ \tangentAS < \tangentS[^2] $ is the space of tangential skew-symmetric 2-tensor fields,
\ie\ $ \thetab = \psi\Eb $ as a local bilinear map.
Levi-Civita compatibility $ \dot{\Eb}=\nullb $ of the tangential material derivative yields the lower-convected rate
\begin{align*}
    \Dlow\thetab
        &= \timeLll\thetab
         = \dot{\psi}\Eb + \psi\dot{\Eb} + \psi\left( \Gb^{T}[\Vb]\Eb + \Eb\Gb[\Vb] \right)
         = \dot{\psi}\Eb +\psi\left( \Eb\Gb[\Vb] - (\Eb\Gb[\Vb])^T \right)\\
        &= \dot{\psi}\Eb + \psi(\Tr\Gb[\Vb])\Eb 
         = \left( \dot{\psi} + \psi\DivC\Vb \right)\Eb
         = \mathring{\psi} \Eb \formComma
\end{align*}
with $ \Gb[\Vb] = \IdS\nablaC\Vb = \nabla\vb - \vnor\shop\in\tangentS[^2] $ the tangential material velocity gradient and $ \shop = -\nablaC\normal \in\tangentS[^2] $ the shape-operator, \resp\ second fundamental form or (extended) Weingarten map.
In conclusion, this demonstrates that the Truesdell derivative defines a natural observer-invariant rate for densities when represented as scalar fields.
An alternative viewpoint on this matter is offered by \cite{Federico_ZfaMuP_2022}, where the Truesdell rate of a scalar field can be interpreted as the forward Piola transform of the time derivative of its backward Piola transform.
The use of the Piola transformation is particularly natural, as it encodes the deformation-induced area change and thereby ensures that the resulting rate is compatible with the transport theorem \eqref{eq:transport_theorem} by construction.
A related idea, following \cite{NitschkeVoigt_JoGaP_2023}, is to introduce a suitable pullback of a future state at time $t+\tau$ onto the surface at time $t$, and then define the time derivative via the resulting difference quotient. 
Roughly speaking, this can be constructed as follows:
Considering the map $ \Phi:\surf\vert_{t} \rightarrow \surf\vert_{t+\tau} $ with $ \Phi\vert_{\tau=0} = \operatorname{id} $, we define the Truesdell pullback of $ \psi\vert_{t+\tau} $ at time $ t $ by
$ \Phi^{*}(\psi\vert_{t+\tau}) 
    :=  \sqrt{|\gb_{\mfrak}|}\vert_{t+\tau} / \sqrt{|\gb_{\mfrak}|}\vert_{t} \; \psi\vert_{t+\tau} $. 
This respects the area deformation by an isotropic stretching factor comprising the determinant of the material metric tensor $ |\gb_{\mfrak}| $, \ie\ Lagrange perspective, at both time steps. 
The Truesdell derivative follows directly as 
$ \mathring{\psi}\vert_{t} 
    = \ddfrac{\tau}\big\vert_{\tau=0}\Phi^{*}(\psi\vert_{t+\tau}) 
    = \lim_{\tau\rightarrow 0}\frac{\Phi^{*}(\psi\vert_{t+\tau}) - \psi\vert_{t}}{\tau}  $ for the material observer and per coordinate transformation for arbitrary observers.

Whichever route we take to define the scalar Truesdell time derivative \eqref{eq:truesdell_rate}, we obtain the following proposition:
\begin{proposition}\label{prop:truesdell_conservation_prop}
    For all $ \qb\in\tangentS $ with $ \int_{\partial\surf} \inner{}{\qb,\nb} dl = 0 $ and co-normal $ \nb $ holds
    \begin{empheq}[box=\fbox]{align}
        \mathring{\psi} 
            &= -\div\qb: \hspace{3em}
          \ddt \int_{\surf}\psi\dS
            = 0 \formPeriod 
    \end{empheq}
\end{proposition}
\begin{proof}
    This follows by applying the transport \eqref{eq:transport_theorem} and Gauss' theorem.
\end{proof}

\section{($L^2,\HmOne$) - Surface Gradient Flows} \label{sec:SGF}

\subsection{General Setting}

We consider a surface energy $ \potenergy = \potenergy[\para,\psi] $ depending on $ \para$, a parameterization of the surface $\surf$, which serves as a proxy/realization for the surface, and $\psi$, a scalar field defined on $\surf$. The functional derivatives $ \DX\potenergy\in\tangentR $ and $ \Dpsi\potenergy\in\tangentScal $ are defined by variations
\begin{align*}
    \innerH{\tangentR}{\DX\potenergy, \Wb}
        &:= \innerH{\tangentR}{\deltafrac{\potenergy}{\para}, \Wb} \formComma
    &\innerH{\tangentScal}{\Dpsi\potenergy, \phi}
        &:= \innerH{\tangentScal}{\deltafrac{\potenergy}{\psi}, \phi}
\end{align*}
for all virtual displacements $ \Wb\in\tangentR $ and $ \phi\in\tangentScal $.
It is crucial to note that $ \DX\potenergy $ is not uniquely defined and depends on the choice of the gauge of surface independence,
\ie\, how $ \psi $ depends on $ \para $ a priori, see \cite{NitschkeSadikVoigt_IJoAM_2023}. 

Gradient-flow formulations are generally not invariant under this choice and it can therefore be advantageous to select the gauge of surface independence consistently with the chosen time derivative.
Corresponding to the scalar Truesdell time derivative we use the Truesdell gauge of surface 
independence
\begin{empheq}[box=\fbox]{align}\label{eq:truesdell_gauge}
    \eth_{\Wb}\psi 
        &= -\psi\DivC\Wb
\end{empheq}
to determine $ \DX\potenergy $ uniquely, 
where $ \eth_{\Wb}\psi[\para] = \ddfrac{\epsilon}\big\vert_{\epsilon=0}\psi[\para+\epsilon\Wb] $ is the deformation derivative \cite{NitschkeSadikVoigt_IJoAM_2023}, \resp\ local spatial variation, for $ \psi=\psi[\para] $ depending on the surface $ \surf $ represented by parameterization $ \para $.
We recall that $ \DivC\Wb = \div\wb - \wnor\meanc $ holds for the orthogonal decomposition $ \Wb=\wb+\wnor\normal $. Note that, similar to the scalar Truesdell time derivative, the Truesdell gauge of surface independence is equivalent to the lower-convected gauge of surface independence 
$ \eth^{\flat\flat}_{\Wb}(\psi\Eb) = (\eth_{\Wb}\psi + \psi\DivC\Wb)\Eb = \nullb $ \cite{NitschkeSadikVoigt_IJoAM_2023},
\ie\ $\eth_{\Wb}(\psi E_{ij}) = 0$ on the covariant proxy components.
Under \eqref{eq:truesdell_gauge} and following \cite{NitschkeSadikVoigt_IJoAM_2023}, the spatial partial variation in terms of functional derivatives yields
\begin{align}
    \innerH{\tangentR}{\frac{\partial\potenergy}{\partial\para}, \Wb}
        &= \innerH{\tangentR}{\DX\potenergy, \Wb} - \innerH{\tangentScal}{\Dpsi\potenergy, \eth_{\Wb}\psi} \notag\\%\label{eq:partial_variation_X_wo_gauge}\\
        &= \innerH{\tangentR}{\DX\potenergy, \Wb} + \innerH{\tangentScal}{\Dpsi\potenergy, \psi\DivC\Wb} \formPeriod \label{eq:partial_variation_X}
\end{align}
Note that there is not any a priori dependency of $ \psi $ on $ \para $, 
\ie\, it holds 
\begin{align}\label{eq:partial_variation_psi}
    \innerH{\tangentScal}{\frac{\partial\potenergy}{\partial\psi}, \phi} 
        &= \innerH{\tangentScal}{\Dpsi\potenergy, \phi} \formPeriod
\end{align}

We here consider the $ (L^2, \HmOne) $-surface gradient flow
\begin{empheq}[box=\fbox]{align}\label{eq:gradient_flow}
    \Mspat \Vb
        &= -\DX\potenergy\formComma
    &\Mdens\mathring{\psi}
        &= \Delta\Dpsi\potenergy \formComma
\end{empheq}
with $ \Delta = \div\circ\nabla: \tangentScal\rightarrow\tangentScal $ the Laplace-Beltrami operator,
suitable boundary conditions, and $ \Mspat, \Mdens > 0$ immobility coefficients.
Or equivalently, in terms of a flux vector $ \qb\in\tangentS $, this surface gradient flow reads
\begin{align}\label{eq:gradient_flow_with_q}
    \Mspat \Vb
        &= -\DX\potenergy\formComma
    &\Mdens\qb
        &= -\nabla\Dpsi\potenergy \formComma
    &\mathring{\psi}
        &= -\div\qb \formPeriod   
\end{align}
Note, that $\Mspat \Vb = -\DX\potenergy$ in general contains tangential and normal components. We will further elaborate on this for special choices of $\potenergy$ below. 
\begin{proposition}\label{prop:props_of_gradient_flow}
    The $ (L^2, \HmOne) $-gradient flow \eqref{eq:gradient_flow}, \resp\ \eqref{eq:gradient_flow_with_q}, \wrt\ the Truesdell gauge of surface independence \eqref{eq:truesdell_gauge}, 
    ensures the conservation property 
    \begin{empheq}[box=\fbox]{align*}
        \ddt \int_{\surf}\psi\dS
             &= 0 \formComma
    \end{empheq}
    and a proper energy dissipation
    \begin{empheq}[box=\fbox]{align}
        \ddt\potenergy \label{eq:energy_rate}
            &= -\left( \Mspat^{-1} \normHsq{\tangentR}{\DX\potenergy} 
                      +\Mdens^{-1} \normHsq{\tangentS}{\nabla\Dpsi\potenergy}\right) \\
            &= -\left( \Mspat \normHsq{\tangentR}{\Vb} 
                      +\Mdens \normHsq{\tangentS}{\qb} \right)
             \le 0\formPeriod
    \end{empheq}
\end{proposition}
\begin{proof}
    The conservation property is already given by \cref{prop:truesdell_conservation_prop}.
    Chain rule, partial variations \eqref{eq:partial_variation_X} and \eqref{eq:partial_variation_psi}, gradient flow \eqref{eq:gradient_flow}, \resp\ \eqref{eq:gradient_flow_with_q}, 
    and integration by parts with convenient boundary conditions, yield
    \begin{align}
        \ddt\potenergy
            &= \innerH{\tangentR}{\frac{\partial\potenergy}{\partial\para}, \Vb}
                    + \innerH{\tangentScal}{\frac{\partial\potenergy}{\partial\psi}, \dot{\psi}}
             =  \innerH{\tangentR}{\DX\potenergy, \Vb}
                    + \innerH{\tangentScal}{\Dpsi\potenergy, \dot{\psi} + \psi\DivC\Vb}\notag\\
            &= \innerH{\tangentR}{\DX\potenergy, \Vb}
                    + \innerH{\tangentScal}{\Dpsi\potenergy, \mathring{\psi}} \label{eq:ddtU_1}\\
%            &= - \Mspat^{-1} \normHsq{\tangentR}{\DX\potenergy}
%               + \Mdens^{-1} \innerH{\tangentScal}{\Dpsi\potenergy, \Delta\Dpsi\potenergy}\notag\\
            &= -\left( \Mspat^{-1} \normHsq{\tangentR}{\DX\potenergy} 
                      +\Mdens^{-1} \normHsq{\tangentS}{\nabla\Dpsi\potenergy}\right) \formPeriod\notag
    \end{align}
\end{proof}
%\begin{remark}
%    In \cref{sec:gradientflow_normal}, we consider the ($L^2,\HmOne$)-surface gradient flow associated with a purely normal evolution of the surface. We demonstrate that this flow can be obtained by imposing the constraint $\vb = \nullb$ on \eqref{eq:gradient_flow}, or equivalently, by a variation of the energy $ \potenergy $ solely in normal direction.
%\end{remark}

\subsection{Surface Tension Flows}\label{sec:energies}

Having established the general surface gradient flow framework in the previous section, we now turn to a specific class of flows, namely surface tension driven flows. For this purpose we consider the surface energy
\begin{empheq}[box=\fbox]{align}\label{eq:energyHT}
    \energyHT
        &:= \int_{\surf} f(\psi) \dS\formComma
\end{empheq}
where $ f(\psi)\in\tangentScal $ depends solely on the scalar field $ \psi\in\tangentS $, 
\ie\ neither on derivatives of $ \psi $ nor any geometric quantities of $ \surf $. 
As a consequence, the following chain rules hold 
\begin{align}\label{eq:chain_rule_HT}
    \eth_{\Wb}f(\psi)
        &= f'(\psi)\eth_{\Wb}\psi \formComma
    &\nabla f(\psi)
        &= f'(\psi)\nabla\psi \formComma    
\end{align}
where $ f'(\psi) = \partial_\psi f(\psi) $ is valid.
Variation \wrt\ $ \psi $ simply yields $ \Dpsi\energyHT = f'(\psi) $.
With the Truesdell gauge of surface independence $ \eth_{\Wb}\psi = -\psi\DivC\Wb $ \eqref{eq:truesdell_gauge},
chain rule \eqref{eq:chain_rule_HT},
and sufficient boundary conditions,
spatial variation results in
\begin{align*}
    \innerH{\tangentR}{\DX\energyHT,\Wb} 
        &= \int_{\surf} \eth_{\Wb}f(\psi) + f(\psi)\DivC\Wb\dS
         = \innerH{\tangentScal}{f'(\psi), \eth_{\Wb}\psi}
           +\innerH{\tangentScal}{f(\psi), \DivC\Wb} \\
        &= \innerH{\tangentScal}{f(\psi) - \psi f'(\psi), \DivC\Wb}
\end{align*}
for all virtual displacements $ \Wb\in\tangentR $.
Therefore, sufficient boundary conditions yield
\begin{align*}
    \DX\energyHT = -\GradC( f(\psi) - \psi f'(\psi)) \formPeriod
\end{align*}
Note that for all scalar fields $ \phi\in\tangentScal $ hold $ \GradC \phi = \DivC(\phi\IdS) $, which allows to rewrite this into a pure isotropic stress formulation.
We would also like to recall that the orthogonal decomposition $ \GradC\phi = \nabla\phi + \phi\meanc\normal\formPeriod $ applies.
Therefore, using the chain rule \eqref{eq:chain_rule_HT}, the tangential part of $ \DX\energyHT $ becomes
\begin{align*}
    \IdS\DX\energyHT 
        &= \psi\nabla f'(\psi) + f'(\psi)\nabla\psi - f'(\psi)\nabla\psi 
         = \psi\nabla f'(\psi)
         = \psi f''(\psi) \nabla\psi\formPeriod
\end{align*}
%We observe that $ \IdS\DX\energyHT = \psi\nabla\Dpsi\energyHT $ holds, as expected also for energies more general than \eqref{eq:energyHT}.
The normal part of $ \DX\energyHT $ reads
\begin{align*}
    \normal\DX\energyHT
        &= \left( \psi f'(\psi) - f(\psi) \right)\meanc\formPeriod
\end{align*}
We summarize the functional derivatives:
\begin{empheq}[box=\fbox]{align}\label{eq:D_HT}
\begin{aligned}
    \Dpsi\energyHT
        &= f'(\psi) \formComma
    &&&\DX\energyHT
        &= -\GradC( f(\psi) - \psi f'(\psi)) 
         = \psi f''(\psi) \nabla\psi
            +\left( \psi f'(\psi) - f(\psi) \right)\meanc\normal \formPeriod
\end{aligned}
\end{empheq}
Eventually, the surface gradient flow \eqref{eq:gradient_flow} reads
\begin{empheq}[box=\fbox]{align}\label{eq:S_gradientflow}
\begin{aligned}
    \Mspat \Vb
        &= \GradC\left( f(\psi) - \psi f'(\psi) \right)\formComma
    &&\text{\resp}\quad
    \left\{\begin{aligned}
        \Mspat\vb 
            &= -\psi f''(\psi)\nabla\psi \formComma\\
         \Mspat \vnor
            &= \left( f(\psi) - \psi f'(\psi)\right)\meanc
    \end{aligned}\right\}\formComma\\
    \Mdens\mathring{\psi}
        &= \div(f''(\psi)\nabla\psi)\formComma
    &&\text{\resp}\quad
    \Mdens\mathring{\psi}
        = f''(\psi)\Delta\psi + f'''(\psi)\normsq{}{\nabla\psi}\formPeriod
\end{aligned}
\end{empheq}
The resulting system of equations couples a geometric evolution equation for the evolution of $\surf$ in normal directions, an equation for tangential movement and a scalar-valued surface partial differential equation on the evolving surface. Alternatively, it may be advantageous, for instance in a numerical implementation, to eliminate the material tangential velocity by substitution.
Inserting the tangential velocity obtained from the surface gradient flow \eqref{eq:S_gradientflow} into the scalar Truesdell time derivative yields
\begin{align}\label{eq:truesdell_velosubs}
    \mathring{\psi}
        &= \mathring{\psi}\vert_{\vb=\nullb}
           -\Mspat^{-1} \div\left( \psi^2 f''(\psi)\nabla\psi \right)\formComma
    &\text{where}\quad
    \mathring{\psi}\vert_{\vb=\nullb}
        &= \partial_t\psi - \left( \nabla_{\vbobs}\psi + \vnor\psi\meanc \right)
\end{align}
is the Truesdell rate for material evolving only in normal direction,
the tangential observer velocity $ \vbobs $ remaining entirely arbitrary.
Due to this, the surfac gradient flow \eqref{eq:S_gradientflow} can be written as
\begin{empheq}[box=\fbox]{align}\label{eq:gradient_flow_vsubs}
\begin{aligned}
    \Mspat\vnor 
        &= \left( f(\psi) - \psi f'(\psi)\right)\meanc\formComma\\
    \Mdens\mathring{\psi}\vert_{\vb=\nullb}
        &= \div\left( \left( 1+\frac{\Mdens}{\Mspat}\psi^2 \right) f''(\psi)\nabla\psi \right)\\
        &= \left( 1 + \psi^2\frac{\Mdens}{\Mspat} \right)f''(\psi)\Delta\psi
           +\left( \left( 1 + \psi^2\frac{\Mdens}{\Mspat}\right)f'''(\psi) 
                    +2\psi\frac{\Mdens}{\Mspat}f''(\psi)\right)\normsq{}{\nabla\psi}\formPeriod
\end{aligned}
\end{empheq}

The functional derivatives \eqref{eq:D_HT} and \cref{prop:props_of_gradient_flow} yield the energy rate
\begin{align}
    \ddt\energyS
        &= -\frac{1}{\Mspat}\label{eq:ddt_energySO}
                \int_{\surf} 
                    \left( \psi f'(\psi) -  f(\psi) \right)^2\meanc^2 
                       + \left( \psi^2 + \frac{\Mspat}{\Mdens} \right)( f''(\psi))^2\normsq{}{\nabla\psi} 
                \dS \leq 0 \formPeriod
\end{align}
Hence, energy dissipation for the gradient flows \eqref{eq:S_gradientflow} and \eqref{eq:gradient_flow_vsubs} is guaranteed.

As the spatial evolution is given exclusively by a gradient term, its argument 
\begin{empheq}[box=\fbox]{align}\label{eq:surface_tension}
    \sigma(\psi) 
            &:=   f(\psi) - \psi f'(\psi)
\end{empheq}
can be identified with a surface tension, \resp\ generalized pressure or isotropic stress.
As a consequence, the surface gradient flow \eqref{eq:S_gradientflow} also reads
\begin{empheq}[box=\fbox]{align}\label{eq:gradient_flow_SO_wrt_tension}
\begin{aligned}
    \Mspat \Vb
        &= \GradC \sigma(\psi)\formComma
    &&\text{\resp}\quad
    \left\{\begin{aligned}
        \Mspat\vb 
            &= \sigma'(\psi)\nabla\psi \formComma\\
         \Mspat \vnor
            &= \sigma(\psi)\meanc
    \end{aligned}\right\}\formComma\\
    \Mdens\mathring{\psi}
        &= -\div\left(\frac{\sigma'(\psi)}{\psi}\nabla\psi\right)\formComma
    &&\text{\resp}\quad
    \Mdens\mathring{\psi}
        =  -\frac{\sigma'(\psi)}{\psi}\Delta\psi
           +\frac{\sigma'(\psi)-\psi\sigma''(\psi)}{\psi^2}\normsq{}{\nabla\psi}\formComma
\end{aligned}
\end{empheq}
where we used that $ \sigma'(\psi) = -\psi f''(\psi) $ and $ \sigma''(\psi) = -(f''(\psi) + \psi f'''(\psi))$ hold. 
The energy rate  \eqref{eq:ddt_energySO} in terms of surface tension states
\begin{align*}
    \ddt\energyS
        &= -\frac{1}{\Mspat}
                \int_{\surf} 
                    \sigma(\psi)^2\meanc^2 
                       + \left( 1 + \frac{\Mspat}{\Mdens\psi^2} \right) \sigma'(\psi)^2\normsq{}{\nabla\psi} 
                \dS \leq 0 \formPeriod
\end{align*}

\begin{remark}
    In \cref{ex:surface_tension_flow_normal} (\cref{sec:gradientflow_normal}) we consider the gradient flow formulations \eqref{eq:S_gradientflow} and \eqref{eq:gradient_flow_SO_wrt_tension} restricted to pure normal evolution of the surface.
    This is the situation considered in most of the cited applications in the introduction and has also been addressed in \cite{ABELS2023236,ABELS202314}, where 
    qualitative properties are discussed and short time existence is shown under various assumptions on the functions $f(\psi)$ and $\sigma(\psi)$. 
\end{remark}

We here refrain from such analytical investigations for the problem including tangential flow and solely focus on the conservation and energy dissipation properties considered in \cref{prop:props_of_gradient_flow}. We discuss some readily accessible special cases:

A constant energy density, \resp\ surface tension, $ f(\psi) = \sigma(\psi) \equiv c > 0$ leads to the classical mean curvature flow
\begin{align*}
   \Mspat\vnor
        &= c \meanc\formComma
\end{align*}
where tangential flow vanishes, \ie\ $ \vb = \nullb $, and the initial density is just transported conservatively, \ie\ $ \mathring{\psi}=0 $,
\resp\ $ \dot{\psi} = \vnor\meanc = \frac{c}{\Mspat} \meanc^2 $. For analytical results we refer to \cite{huisken1984flow} and for numerical approaches, see \cite{Deckelnick_Dziuk_Elliott_2005}.

A linear energy density $ f(\psi) = c\psi $, \resp\ vanishing surface tension $ \sigma(\psi)\equiv 0 $, leads to the static state $ \surf = \surf\vert_{t_0} $, \ie\ $ \Vb=\nullb $, and $ \mathring{\psi}=\dot{\psi}=0 $.
This is trivial, since $ \energyS  $ is conserved by construction in that case.

A quadratic energy density $ f(\psi) = \frac{c}{2}\psi^2 $, \resp\ quadratic surface tension $ \sigma(\psi) = -\frac{c}{2}\psi^2 $, leads to a density-weighted mean-curvature flow with conserved density diffusion:
\begin{align}\label{eq:psi2_gradient_flow}
    \Mspat\vb
        &= -c\psi\nabla\psi\formComma
    &\Mspat\vnor
        &= -\frac{c}{2}\psi^2\meanc\formComma
    &\Mdens\mathring{\psi}
        &= c\Delta\psi\formPeriod
\end{align}
As this model provides a particularly simple, though nontrivial, example, 
we take this opportunity to discuss the effect of the choice of the Truesdell gauge of surface independence and time derivative within this setting. Variation from the scratch of the scalar surface energy \eqref{eq:energyHT} for $ f(\psi) = \frac{c}{2}\psi^2 $ in arbitrary displacement directions $ \Wb\in\tangentR $, and considering a different gauge of surface independence, here the material gauge of surface independence ($\eth_{\Wb}\psi=0$, \cite{NitschkeSadikVoigt_IJoAM_2023}) yields
\begin{align*}
    \innerH{\tangentR}{\deltafrac{\energyS}{\para},\Wb}
        &= \frac{c}{2}\int_{\surf} \eth_{\Wb}\psi^2 + \psi^2\DivC\Wb\dS\\
        &= \frac{c}{2}\left(\innerH{\tangentScal}{\psi, \eth_{\Wb}\psi} - \innerH{\tangentR}{\GradC\psi^2,\Wb}\right)\\
        &= \innerH{\tangentR}{\widetilde{\DX}\energyS,\Wb} \formComma
\end{align*}
where $\widetilde{\DX}\energyS = -\frac{c}{2}\GradC\psi^2$ and $ \GradC\psi^2 = 2\psi\nabla\psi + \psi^2\meanc\normal $. This leads to the gradient flow
\begin{align}\label{eq:psi2_gradient_flow_material_gauge}
    \Mspat\vb
        &= c\psi\nabla\psi\formComma
    &\Mspat\vnor
        &= \frac{c}{2}\psi^2\meanc\formComma
    &\Mdens\mathring{\psi}
        &= c\Delta\psi\formComma
\end{align}
\ie\, the spatial forces are in opposite direction contrarily to the gradient flow \eqref{eq:psi2_gradient_flow}, where we use the Truesdell gauge of surface independence ($\eth_{\Wb}\psi=-\psi\DivC\Wb$). Calculating the energy rate with the transport theorem \wrt\ the scalar Truesdell time derivative \eqref{eq:truesdell_rate} reveals
\begin{align*}
    \ddt\energyS 
        &= c\left( \innerH{\tangentScal}{\psi,\dot{\psi}} + \frac{1}{2}\innerH{\tangentScal}{\psi^2,\DivC\Vb} \right)
         = c\left( \innerH{\tangentScal}{\psi,\mathring{\psi}} - \frac{1}{2}\innerH{\tangentScal}{\psi^2,\DivC\Vb} \right)\\
        &= c\left( \innerH{\tangentScal}{\psi,\mathring{\psi}} + \frac{1}{2}\innerH{\tangentR}{\GradC\psi^2,\Vb} \right)\\
        &= c\left( \innerH{\tangentScal}{\psi,\mathring{\psi}} + \innerH{\tangentS}{\psi\nabla\psi,\vb} + \frac{1}{2}\innerH{\tangentScal}{\psi^2\meanc,\vnor} \right)\formComma
\end{align*}
which is consistent with \eqref{eq:ddtU_1}.
Substitution of the gradient flows into this energy rate yields
\begin{align*}
    \ddt\energyS
        &= -\frac{c^2}{\Mdens}\normHsq{\tangentS}{\nabla\psi}
           + \begin{cases}
                -\frac{c^2}{\Mspat}\left( \normHsq{\tangentS}{\psi\nabla\psi} + \frac{1}{4}\normHsq{\tangentScal}{\psi^2\meanc} \right)\formComma
                        &\text{for \eqref{eq:psi2_gradient_flow},} \\
                +\frac{c^2}{\Mspat}\left( \normHsq{\tangentS}{\psi\nabla\psi} + \frac{1}{4}\normHsq{\tangentScal}{\psi^2\meanc} \right)\formComma
                        &\text{for \eqref{eq:psi2_gradient_flow_material_gauge},}
             \end{cases}
\end{align*}
where the first case is consistent to \cref{prop:props_of_gradient_flow} and gradient flow \eqref{eq:psi2_gradient_flow} ensures a decreasing energy.
Contrarily, the surface gradient flow \eqref{eq:psi2_gradient_flow_material_gauge} could increase the energy, especially if $ \Mspat $ is sufficiently smaller than $ \Mdens $.

\begin{remark}
Note that the spatial forces $ -\widetilde{\DX}\energyS $ as well as $  -\DX\energyS $ are solely in terms of $ \GradC $
and the Helmholtz decomposition $ \tangentR = \GradC\tangentScal \oplus \ker(\DivC\vert_{\tangentR})$ holds.
Therefore, in an inextensible setting ($ \DivC\Vb = 0 $), the divergence-free part of both, $ \widetilde{\DX}\energyS $ as well as $  \DX\energyS $, is identically zero. As a consequence, both resulting models would be equal with no surface motion resulting in $\Mdens (\partial_t \psi - \nabla_{\vbobs}\psi) = c \Delta \psi$ and thus reducing to the standard diffusion equation for an Eulerian, \ie\ stationary, observer ($\vbobs = 0$) and leading to an energy rate of $ \ddt\energyS = -\frac{c^2}{\Mdens}\normHsq{\tangentS}{\nabla\psi} $.
\end{remark}

While these examples are either trivial, or considered for illustration purposes only, also physically relevant formulations arise. In the context of two-phase flows with surfactants free energy contributions based on a logarithmic Flory--Huggins type potential are considered. Originally introduced in the context of polymer solutions, this energy can be interpreted more generally as the configurational entropy of mixing on a discrete set of sites 
(see \eg\ \cite{Flory_TJoCP_1942,Huggins_TJoCP_1941,Doi_2013}).
In the present setting, we interpret the scalar field $ \psi $ as the local surface coverage of surfactant molecules, \ie\ $ 0<\psi< 1 $.
The actual surfactant concentration is given by $ \psi_c = \psi_{\infty}\psi $, with $ \psi_{\infty} $ representing the maximum possible concentration on the surface.
The complementary fraction $1 - \psi$ then represents the available free surface area, \ie\ unoccupied adsorption sites. 
Under this interpretation, the surface is viewed as a collection of finitely many equivalent sites that can either be occupied by surfactant molecules or remain empty. 
The number of possible configurations associated with a given coverage $ \psi $ leads, via standard combinatorial arguments \wrt\ a lattice model \cite{Doi_2013}, to the entropy density  
$ \beta(\psi\ln\psi + (1-\psi)\ln(1-\psi)) $ with $\beta > 0 $.
Additionally, we include a purely geometric energy density $ \sigma_0 > 0$, independent of the surfactant concentration.
Finally, we take into account interactions between solute surfactant molecules by introducing an interaction energy density $ \chi(\psi(1-\psi)) $ \cite{Doi_2013} with $\chi\in\R$. 
Combining all contributions, we obtain the energy density
\begin{align}\label{eq:energy_dens_surfactants}
    f(\psi) = \sigma_0 + \beta\left( \psi\ln\psi + (1-\psi)\ln(1-\psi) \right) + \chi\psi(1-\psi)\formPeriod
\end{align}
The parameters $ \beta > 0 $ and $ \chi\in\R $ relate the number of available sites for the surfactants to the surface area in a temperature-dependent manner and in dependence on the effective interaction energy and the coordination number, see \cite{Doi_2013}. 
The logarithmic contribution in \eqref{eq:energy_dens_surfactants} naturally enforces the physically admissible bounds ($ 0<\psi< 1 $), as the free energy diverges when the interface becomes either completely depleted or fully saturated. Note that $ f $ becomes a double-well potential for $ \chi > 2\beta $, \ie\, phase separation is to be expected in that case. The corresponding surface tension \eqref{eq:surface_tension}, yields
\begin{align} \label{eq:energy_dens_surfactants_sigma}
    \sigma(\psi) 
        &= \sigma_0 + \beta \ln(1-\psi) + \chi\psi^2 \formPeriod
\end{align}
In the absence of interaction ($ \chi=0 $), this corresponds to the Langmuir equation of state \cite{LANGMUIR} and is frequently used in two-phase flows with surfactants, \eg, \cite{VelankarZhouEtAl_JoCaIS_2004,SmanGraaf_RA_2006,ErikTeigenSongEtAl_JoCP_2011,EngblomDo-QuangEtAl_CiCP_2013,Barrett2014421,YangJu_CMiAMaE_2017,FRACHON2023111734,vanSluijs_2025}. In this context $\GradC \sigma(\psi) = \DivC (\sigma(\psi) \IdS) $ acts as an interface force and balances the jump of the stresses in the bulk domains. The interface force contains tangential $\sigma'(\psi)\nabla\psi$ and normal $\sigma(\psi) \meanc$ components. In our context the corresponding interface forces determine the tangential and normal velocity and the surface gradient flow \eqref{eq:S_gradientflow}, \resp\ \eqref{eq:gradient_flow_SO_wrt_tension}, reads
\begin{align*} 
\begin{aligned}
    \Mspat \Vb
        &= \GradC \left( \sigma_0 + \beta\ln(1-\psi) + \chi\psi^2 \right)\formComma
    &&\text{\resp}\quad
    \left\{\begin{aligned}
        \Mspat\vb 
            &= \left( 2\chi\psi - \frac{\beta}{1-\psi} \right)\nabla\psi \formComma\\
         \Mspat \vnor
            &= \left(\sigma_0 + \beta\ln(1-\psi) + \chi\psi^2 \right)\meanc
    \end{aligned}\right\}\formComma\\
    \Mdens\mathring{\psi}
        &= \div\left(\left( \frac{\beta}{\psi(1-\psi)} - 2\chi \right)\nabla\psi\right)\formComma
    &&\text{\resp}\quad
    \Mdens\mathring{\psi}
        = \left( \frac{\beta}{\psi(1-\psi)} - 2\chi \right)\Delta\psi - \beta\frac{1-2\psi}{\psi^2(1-\psi)^2}\normsq{}{\nabla\psi}\formPeriod\\
\end{aligned}
\end{align*}
Also the evolution for $\psi$ is mostly considered in this form in the context of two-phase flows with surfactants \cite{ErikTeigenSongEtAl_JoCP_2011}. However, mostly by approximating the diffusion using a constant diffusion coefficient, which does not alter the conservation property. In the context of two-phase flow with surfactants the tangential force is know to lead to the Marangoni effect \cite{marangoni1865sull}. %We will numerically demonstrate that the tangential force also contributes to the evolution in ($L^2, \HmOne$)-surface gradient flows, without considering any bulk domain. 
The derived ($L^2,\HmOne$) - surface gradient flow with $f(\psi)$ given by \eqref{eq:energy_dens_surfactants} demonstrates that the tangential flow also influences the evolution without considering any bulk domain and should be taken into account for a quantitative description of the dynamics. 

We will computationally quantify the differences between the gradient flow formulations \eqref{eq:S_gradientflow} and \eqref{eq:gradient_flow_SO_wrt_tension} and the restricted models with pure normal evolution of the surface in \cref{ex:surface_tension_flow_normal} (\cref{sec:gradientflow_normal}) considering the energy density \eqref{eq:energy_dens_surfactants} or the surface tension \eqref{eq:energy_dens_surfactants_sigma} with $\chi = 0$ in \cref{sec:C}.

\subsection{Surface Tension Flows in Height Observer Formulation}

Computationally more trackable than the previously considered formulations are formulations using a height observer. We therefor also formulate the equations for surface tension flows in this setting. Following the notation and results in \cref{sec:height_formulation_setting} for a height observer, the surface gradient flow \eqref{eq:S_gradientflow} reads: 

\noindent
Find the covariant material tangential velocity proxy field $ \vb^{\flat}=[v_x,v_y]$, height field $ h\in\tangentScal $ and density field $ \psi\in\tangentScal $ \soth
\begin{empheq}[box=\fbox]{align}\label{eq:gradient_flow_SO_height}
\begin{aligned} 
        \Mspat \vb^{\flat}
            &= -\psi f''(\psi)\partial\psi\formComma\\
        \Mspat \partial_t h
            &= |g|\left(  f(\psi) - \psi f'(\psi)\right)\hfrak\formComma\\
        \Mdens\left(\partial_t \psi - \left( \psi\hfrak+\frac{\partial\psi\cdot\dh}{|g|} \right)\partial_t h
                    +\psi\left( \partial\cdot\vF - \frac{\dh\cdot\partial\vF\cdot\dh}{|g|}\right)
                    +\vF\cdot\left( \partial \psi - \left( \psi\Delta h+\frac{\partial\psi\cdot\dh}{|g|} \right)\dh \right) \right)\hspace{-30em}\\
            &=  f''(\psi)
                  \left(\partial\cdot\partial\psi - (\partial\psi \cdot\dh) \hfrak - \frac{\dh\cdot\partial^2\psi\cdot\dh}{|g|} \right)
                +  f'''(\psi)\left( \partial\psi\cdot\partial\psi - \frac{(\partial\psi\cdot\dh)^2}{|g|} \right)
    \end{aligned}
\end{empheq}
    holds, where
\begin{align*}
    \hfrak
        &= \frac{\meanc}{\sqrt{|g|}}
         =\frac{\partial\cdot\dh}{|g|} - \frac{\dh\cdot\ddh\cdot\dh}{|g|^2} \formComma &|g| &= 1 + (\dh)\cdot(\dh) \formPeriod
\end{align*}
While the height observer velocity is simply given by $ \Vbobs=[0,0,\partial_t h]^T $, the material velocity is obtained as
    \begin{align*}
        \Vb
            &= \vb +\vnor\normal
             =   \begin{bmatrix}
                     \vF \\ 0
                 \end{bmatrix}
                -\frac{\partial_t h + \vF \cdot \dh}{|g|} 
                    \begin{bmatrix}
                        \dh \\ -1
                    \end{bmatrix}\formPeriod
    \end{align*}
Considering instead \eqref{eq:gradient_flow_vsubs} and substituting $ \partial_t h $ we get: 

\noindent
Find height field $ h\in\tangentScal $ and density field $ \psi\in\tangentScal $ \soth
\begin{empheq}[box=\fbox]{align}\label{eq:gradient_flow_SO_new_height}
\begin{aligned}
    \Mspat \partial_t h
        &= |g|\left( f(\psi) - \psi f'(\psi)\right)\hfrak\formComma \\
     \Mdens\partial_t \psi
        &= \left( 1 + \psi^2\frac{\Mdens}{\Mspat} \right)f''(\psi)
                        \left(\partial\cdot\partial\psi - \frac{\dh\cdot\partial^2\psi\cdot\dh}{|g|} \right)\\
    &\quad +\left( \left( 1 + \psi^2\frac{\Mdens}{\Mspat}\right)f'''(\psi) 
                    +2\psi\frac{\Mdens}{\Mspat}f''(\psi)\right)
                        \left( \partial\psi\cdot\partial\psi - \frac{(\partial\psi\cdot\dh)^2}{|g|} \right)\\
    &\quad +\left( \frac{\Mdens}{\Mspat}\left( f(\psi) - \psi f'(\psi) - \psi^2 f''(\psi)\right) - f''(\psi)\right)
                                    (\partial\psi\cdot\dh)\hfrak\\
    &\quad + |g|\psi\frac{\Mdens}{\Mspat}\left( f(\psi) - \psi f'(\psi)\right)\hfrak^2
\end{aligned}
\end{empheq}
holds. The material velocity is obtained as
\begin{align*}
    \Vb
        &= -\frac{\psi f''(\psi)}{\Mspat}\nabla\psi +  \frac{\partial_t h}{\sqrt{|g|}}\normal
         = \frac{1}{|g|}\begin{bmatrix}
                            -\frac{|g|\psi f''(\psi)}{\Mspat}\partial\psi - \left( \partial_t h - \frac{\psi f''(\psi)}{\Mspat}(\partial\psi\cdot\dh) \right)\dh \\
                            \partial_t h - \frac{\psi f''(\psi)}{\Mspat}(\partial\psi\cdot\dh)
                        \end{bmatrix}\formPeriod
\end{align*}

Note that these formulations do not consider a small displacement assumption $\partial h \ll 1$, which is often used in height observer formulations. %The formulations provide an easy to implement form to numerically explore the effects of the various terms. However, here we only relate these formulations to simplified situations. 

Before we solve these formulations numerically and explore the effects of the various terms, lets also relate these formulations to simplified situations. 

For a constant surface tension $f(\psi) = c > 0$ the formulations lead to
\begin{align*}
   \Mspat\partial_t h   
        &= |g| c \hfrak\formComma 
\end{align*}
which can be rewritten in the more commonly used form for mean curvature flow in height observer formulation 
\begin{align*}
   \Mspat \frac{\partial_t h}{\sqrt{|g|}}    
        &= c \partial \cdot \left( \frac{\dh}{\sqrt{|g|}} \right)\formPeriod 
\end{align*}
Also in this case tangential flow vanishes, \ie\ $ \vb = \nullb $, and the initial density is transported, \ie\ $ \mathring{\psi}=0 $. For results on mean curvature flow in height observer formulations we refer to \cite{Deckelnick_Dziuk_Elliott_2005}. Other special cases, as discussed in the previous sections, can be formulated accordingly.

\section{Numerical Results and Conclusions} \label{sec:C}

We consider the derived height observer formulations with periodic boundary conditions and solve the equations numerically using a semi-implicit pseudospectral method. In order to explore the impact of the tangential velocity we compare the full problem \eqref{eq:gradient_flow_SO_height} with the restricted model with pure normal evolution of the surface \eqref{eq:gradient_flow_SO_height_normal} in \cref{ex:surface_tension_flow_normal} (\cref{sec:gradientflow_normal}). We consider the energy density \eqref{eq:energy_dens_surfactants} with $\sigma_0 = 1$, $\beta = 0.75$ and $\chi = 0$ and use parameters $M_{\mathbf X} = 5$ and $M_\psi = 1$. The computational domain is $(2\pi, 2\pi)$ with 128 discretisation points in each direction. Only a quarter of this domain is considered in \cref{fig:results}. The time step-size is $10^{-5}$. We consider the initial conditions $h_{init}(x,y) = \sin{2 x} \sin{2y}$ and $\psi_{init} = 0.25$. In \cref{fig:results} a) - c) basic properties of both models are confirmed, such as energy dissipation and conservation of $\psi$. Comparing these quantities already demonstrate the differences in the evolution between both models. The equilibrium state is reached faster for the full model. Furthermore, the emerging differences in $\psi$ are smaller for the full model. In \cref{fig:results} d) the initial height profile is shown together with a cut-plane indicated by a line, along we measure various quantities over time. In \cref{fig:results} e) and f) the solutions at $t = 0.2$ are shown for the full and restricted problem, respectively. The shape and the density $\psi$ are shown indicating a flattened morphology and anisotropic $\psi$ distribution. In \cref{fig:results} e) also the tangential velocity is shown which drives additional material flux away from the maxima of $\psi$. While differences in the shape are barely visible between both models, the density profile clearly differs as a result of the additional tangential relaxation. Such differences become most apparent in \cref{fig:results} g), h) and i), showing the solutions along the indicated cut-plane in \cref{fig:results} d) at various time instances. The largest difference becomes visible for the velocity components in tangential and normal direction $\Vb = \vb +\vnor\normal$. Basic numerical tests have been considered to ensure that these results are properties of the model and not numerical artefacts. These tests (not shown) demonstrate linear convergence of the conservation error with respect to time stepping and the same numerical energy dissipation rate as for the continuous model. 

We conclude the importance of tangential flows even in simple surface tension flows if the energy density depends on a scalar surface quantity. The tangential flow is induced by gradients of this scalar surface quantity and has an impact on the evolution of the surface. In the considered setting it provides an additional relaxation mechanism and leads to a faster energy decay in the ($L^2, \HmOne$) - surface gradient flow. The scalar Truesdell time derivatives together with the Truesdall gauge of surface independence provide useful tools to construct such flows also for more general energies and guarantee energy dissipation and conservation properties.

\begin{figure}[p]
\includegraphics*[angle = -0, width = 1 \textwidth ]{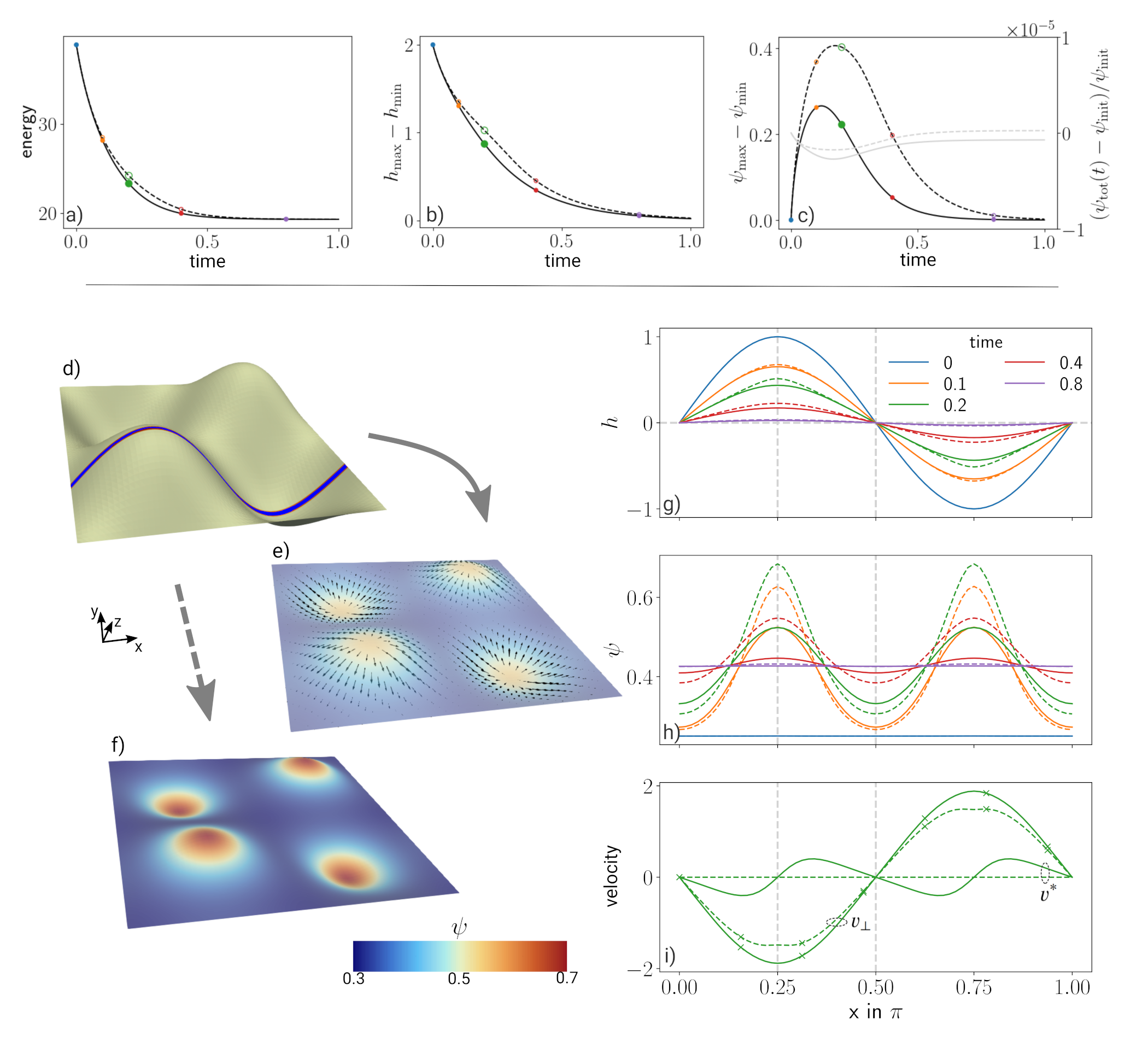}
\caption{Numerical results for surface relaxation. a)-c) Evolution of
  characteristic system properties: a) Energy decay, b) Surface flattening via
  maximum height difference, and c) Extrema of $\psi$ (black) with total mass
  conservation (gray). Symbols indicate time instances analyzed in
  h)-i). Solid lines correspond to the full model
  \eqref{eq:gradient_flow_SO_height}, dashed lines to the restricted model
  \eqref{eq:gradient_flow_SO_height_normal}. d) Initial surface configuration;
  the blue line indicates the cross-section used for analysis in g)-i) and
  $\psi$ is constant. e) Solution at $t=0.2$ for the full model with indicated
  tangential flow (black arrows). f) Solution at $t=0.2$ for restricted
  model. g) and h) Profiles of height $h$ and density $\psi$ for $t \in \{0,
  0.1, 0.2, 0.4, 0.8\}$, respectively. i) Tangential $v^* = \text{sgn}{(v_x)}
  ||\vb||$ and normal $\vnor$ velocity components of $\Vb$ at $t=0.2$.
\label{fig:results}}
\end{figure}

\appendix
\section*{Appendix}

\section{Height Observer Formulation} \label{sec:height_formulation_setting}

We consider a height observer given  by  parameterization \cite{BenoitMarechalNitschkeEtAl_MoM_2024}
\begin{empheq}[box=\fbox]{align*}
    \para_{\ofrak}:\ 
        (t,x,y)\mapsto \para_{\ofrak}(t,x,y)
        &:= \left[ x, y, h(t,x,y) \right]^T \in\surf\subset\R^3\formComma
\end{empheq} 
where $ h\in\tangentScal $ is the scalar height field, $ t\in\mathcal{T}\subseteq\R $ the time coordinate, and $ (x,y)\in\mathcal{U}\subseteq\R^2 $ local coordinates,
\ie\ $ (\surf, \para\vert_{t}^{-1}) $ is a time-dependent local chart.
The observer frame and normal field yield
\begin{align*}
    \partial_x \para_{\ofrak}
        &= \left[ 1, 0, \partial_x h  \right]^T\formComma
    &\partial_y \para_{\ofrak}
        &= \left[ 0, 1, \partial_y h  \right]^T\formComma
    &\normal
        &= \frac{\left[-\partial_x h, -\partial_y h, 1\right]^T}{\sqrt{|g|}}\formComma
\end{align*} 
where
\begin{align*}
    |g|
        &:= g_{11}g_{22} - 2g_{12}
          = 1 + (\partial_x h)^2 + (\partial_y h)^2 \formComma
    &g_{ij}
        &= \inner{}{\partial_i \para_{\ofrak}, \partial_j \para_{\ofrak}} 
         = \delta_{ij} + (\partial_i h)(\partial_j h) \formPeriod
\end{align*}
Moreover, the inverse observer metric stated
$
    g^{ij}
        = \delta^{ij} - \frac{(\partial_i h)(\partial_j h)}{|g|}\formComma
$
as a consequence of $ g_{ik}g^{kj} = \delta_i^j $.
Since we intend to employ a Cartesian $ \R^2 $ calculus on $ \mathcal{U} $, we define the partial differential operator
\begin{align*}
    \partial \bullet
        &= (\partial_x \bullet)\otimes\eb_x + (\partial_y \bullet)\otimes\eb_y \formComma
\end{align*}
where $ \bullet $ could be anything convenient, \eg\ a scalar or a matrix field, 
and $ \{\eb_x, \eb_y \} $ is a Cartesian frame on $ \mathcal{U}\subseteq\R^2 $.
For instance, we could write 
\begin{align*}
    |g| &= 1 + (\dh)\cdot(\dh)
\end{align*}
or for scalar fields $ f\in\tangentScal $
\begin{align*}
    \partial\cdot\partial f
        &= \partial_x^2 f + \partial_y^2 f \formPeriod
\end{align*}
We use musical symbols $ \flat $ and $ \sharp $ for local covariant and contravariant proxy notations,
\eg\ the covariant and contravariant metric tensor proxy reads
\begin{align*}
    \gb^{\flat}
        &= \IdTwo + \dh \otimes \dh \formComma
    &\gb^{\sharp}
        &= (\gb^{\flat})^{-1}
         = \IdTwo -\frac{\dh \otimes \dh}{|g|}\formComma
    &\IdTwo %= \gb^{\flat\sharp} = \gb^{\sharp\flat} 
        &= \gb^{\flat}\cdot\gb^{\sharp}
         = \gb^{\sharp}\cdot\gb^{\flat}\formComma
\end{align*} 
or for tangential vector fields $ \wb\in\tangentS $ holds
\begin{align*}
    \wF
        &= \left[ w_x, w_y \right]
         = \gb^{\flat} \cdot \wS
         = \wS + \left(  \wS \cdot \dh \right) \dh\formComma
    \hspace{4em}\wS
        = \left[ w^x, w^y \right]^T
         = \wF \cdot \gb^{\sharp}
         = \wF - \frac{\wF \cdot \dh}{|g|}\dh \formComma\\
    \wb &= w^x\partial_{x}\para_{\ofrak} + w^y\partial_{y}\para_{\ofrak}
         = \partial\para_{\ofrak} \cdot \wS
         = \begin{bmatrix}
                \wS \\ \wS\cdot\dh
            \end{bmatrix}
         = \partial\para_{\ofrak} \cdot \wF \cdot \gb^{\sharp}
         = \begin{bmatrix}
            \wF \\ 0
           \end{bmatrix}
           +\frac{\wF \cdot \dh}{\sqrt{|g|}}\normal\formComma 
\end{align*}
where $  \partial\para_{\ofrak} = [\IdTwo,  \dh]^{T} $,
\ie\, it also hold the relations $ \wF = \wb\cdot\partial\para_{\ofrak} $ and $ \gb^{\flat} = (\partial\para_{\ofrak})^T \cdot \partial\para_{\ofrak} $.
The local inner product is evaluated in the standard manner via contraction of proxies with differing variances,
\eg\, for tangential fields $ \wb,\etab\in\tangentS $ by
\begin{align*}
    \inner{}{\wb, \etab}
        &= \wb^{\sharp}\cdot\etab^{\flat}
         = \wb^{\flat}\cdot\etab^{\sharp}
         =  \wb^{\flat}\cdot\etab^{\flat} - \frac{(\wb^{\flat}\cdot\dh)(\etab^{\flat}\cdot\dh)}{|g|}\formComma
    &\normsq{}{\wb}
        &= \inner{}{\wb, \wb}
         = \wb^{\flat}\cdot\wb^{\flat} - \frac{(\wb^{\flat}\cdot\dh)^2}{|g|} \formPeriod
\end{align*}
We use only one musical symbol if the associated proxy is fully co- or contravariant.
If the proxy is mixed co- and contravariant, then we use one musical symbol for every column dimension,
\eg\ $ \gb^{\sharp\flat} = \gb^{\flat\sharp} = \IdTwo $.
The Gauss-Weingarten equation  reads
\begin{align*}
    (\partial^2\para_{\ofrak})^{T_{(1\,2\,3)}}
        &= \Gammab^{\flat}\cdot\gb^{\sharp}\cdot(\partial\para_{\ofrak})^{T}
          +\shop^{\flat}\otimes\normal\formComma
\end{align*}
where
\begin{align*}
    \partial^2\para_{\ofrak} 
        &= [0,0,1]^T \otimes \ddh\formComma
    &(\partial^2\para_{\ofrak})^{T_{(1\,2\,3)}}
        &= \ddh \otimes [0,0,1]^T \formComma
\end{align*}
the Christoffel symbols (of first kind) as matrix proxy is
\begin{align*}
    \Gammab^{\flat}
        &=  (\partial^2\para_{\ofrak})^{T_{(1\,2\,3)}} \cdot \partial\para_{\ofrak}
         = \ddh \otimes \dh \formComma
\end{align*}
and the covariant proxy of the second fundamental form/shape operator is
\begin{align*}
    \shop^{\flat}
        &= (\partial^2\para_{\ofrak})^{T_{(1\,2\,3)}} \cdot \normal
         = \frac{1}{\sqrt{|g|}} \ddh \formPeriod
\end{align*}
As a consequence, the mean curvature is
\begin{align}\label{eq:meanc_gaussc}
    \meanc 
        &= \shop^{\flat}\dbdot\gb^{\sharp}
         = \frac{1}{\sqrt{|g|}}\left( \partial\cdot\dh - \frac{\dh \cdot \ddh \cdot \dh}{|g|} \right)\formComma
\end{align}
where we used that the determinant $ |\bullet| $ is a multiplicative map on square matrices,
and that with $ |g|=|\gb^{\flat}| $ and $ |\gb^{\sharp}| = |\gb^{\sharp}\gb^{\flat}\gb^{\sharp}| = |g||\gb^{\sharp}|^2 $
follows $ |\gb^{\sharp}| = |g|^{-1} $.
Since $ \dh \cdot \gb^{\sharp} = |g|^{-1}\dh $ holds,
the Christoffel symbols of second kind yield in matrix proxy notation
\begin{align}\label{eq:ch_2}
    \Gammab^{\flat\flat\sharp}
        &= \Gammab^{\flat}\cdot\gb^{\sharp}
         = \frac{1}{|g|}\Gammab^{\flat}
         = \frac{1}{|g|}\ddh \otimes \dh \formPeriod
\end{align}

Based on the preceding considerations, we are now in a position to formulate differential operators based on the covariant derivative $ \nabla $
in the proxy notation described above.
Scalar fields $ f\in\tangentScal $, for instance, yield
\begin{align*}
    (\nabla f)^{\flat}
        &= \partial f\formComma
    \quad\text{\resp}\quad(\nabla f)^{\sharp} 
        =\partial f - (\partial f \cdot \dh)(\nabla h)^{\sharp}\formComma
    &&\text{where}\quad
     &(\nabla h)^{\sharp}
        &=\frac{1}{|g|}\dh\formComma\\
    (\nabla^2 f)^{\flat}
        &=\partial^2 f - (\partial f \cdot \dh)(\nabla^2 h)^{\flat} \formComma
    &&\text{where}\quad
     &(\nabla^2 h)^{\flat}
        &= \frac{1}{|g|}\ddh\formComma\\
   \Delta f
       &= \partial\cdot\partial f - \frac{\dh \cdot \partial^2 f \cdot \dh}{|g|} - (\partial f \cdot \dh) \Delta h\formComma
   &&\text{where}\quad
    &\Delta h
       &= \frac{1}{|g|}\left( \partial\cdot\dh - \frac{\dh \cdot \ddh \cdot \dh}{|g|} \right)\formComma 
\end{align*}
since $ (\nabla f)^{\sharp} = (\nabla f)^{\flat}\cdot\gb^{\sharp} $,
$ (\nabla^2 f)^{\flat} = \partial^2 f - \Gammab^{\flat\flat\sharp}\cdot\partial f $,
and $ \Delta f = \div\nabla f = (\nabla^2 f)^{\flat}\dbdot\gb^{\sharp} $ is valid.
We also observe that 
\begin{align*}
    \shop 
        &= \sqrt{|g|}(\nabla^2 h)\formComma
    &\meanc
        &= \sqrt{|g|}\Delta h
\end{align*}
hold.
For tangential vector fields $ \wb\in\tangentS $ and vector fields $ \Wb=\wb+\wnor\normal\in\tangentR $ with scalar normal component $ \wnor\in\tangentScal $,
we obtain
\begin{align*}
    (\nabla\wb)^{\flat} 
        &= \partial \wF - (\wF \cdot \dh)(\nabla^2 h)^{\flat}\formComma\\
    \div\wb
        &= \partial\cdot\wF - \frac{\dh \cdot \partial \wF \cdot \dh}{|g|} - (\wF \cdot \dh) \Delta h\formComma\\
    \DivC\Wb
        &= \partial\cdot\wF - \frac{\dh \cdot \partial \wF \cdot \dh}{|g|} 
            - (\wF \cdot \dh + \sqrt{|g|}\wnor) \Delta h\formComma
\end{align*}
since $(\nabla\wb)^{\flat} = \partial \wF - \Gammab^{\flat\flat\sharp}\cdot\wF$, 
$ \div\wb = (\nabla\wb)^{\flat} \dbdot \gb^{\sharp} $,
and $ \DivC\Wb = \div\wb - \wnor\meanc $ hold.

The above considerations pertain to a purely instantaneous framework. 
In the following, we extend the analysis to include dynamical quantities.
The observer velocity is given by
\begin{align*}
    \Vbobs
        &= \partial_t\para_{\ofrak}
         = \left[ 0, 0, \partial_t h \right]^T \in\tangentR\formPeriod
\end{align*}
The normal velocity is equal for all observer, including the material, and reads
\begin{align*}
    \vnor 
        &= \Vbobs\cdot\normal
         = \frac{\partial_t h}{\sqrt{|g|}} \in \tangentScal\formPeriod 
\end{align*}
Note that in our models the normal spatial equations are always given by $ \vnor = -\phi $, with $ \phi\in\tangentScal $, which thus translates to
\begin{align*}
    \partial_t h = -\sqrt{|g|}\phi\formPeriod
\end{align*}
The tangential observer velocity yields
\begin{align*}
    \vbobs^{\flat}
        &= \Vbobs \cdot \partial\para_{\ofrak}
         = (\partial_t h) \dh\formComma
    &\vbobs^{\sharp}
         &= \vbobs^{\flat} \cdot \gb^{\sharp}
          = \frac{\partial_t h}{|g|}\dh\formComma
     &\vbobs
         &=  \frac{\partial_t h}{|g|}
                \begin{bmatrix}
                    \dh \\ |g| - 1
                \end{bmatrix}\formPeriod
\end{align*}
While the normal material velocity is fixed by $ \vnor $, the tangential material velocity $ \vb\in\tangentS $ remains undefined by the observer.
As a consequence we could use it as a degree of freedom. 
The choice of proxy for $ \vb $ is not unique.
However, we decide to use the covariant proxy 
\begin{empheq}[box=\fbox]{align*}
    \vF &:= [v_x, v_y] \formComma
\end{empheq}
since tangential spatial equations reads $ \vb = -\psi\nabla\mu $ in this paper, where $ \psi,\mu\in\tangentScal $, 
which simply results in
\begin{align*}
    \vF = -\psi\partial\mu\formComma
\end{align*}
\ie\ $ v_x = -\psi\partial_x\mu $ and $ v_y = -\psi\partial_y\mu $.
For visualizations alone, it would be helpful to have the material velocity at the surface.
For this purpose the tangential and full material velocity in terms of $ \vF $ read
\begin{empheq}[box=\fbox]{align*}
    \vb
        &=  \begin{bmatrix}
                \vF \\ 0
            \end{bmatrix}
           +\frac{\vF \cdot \dh}{\sqrt{|g|}}\normal \in\tangentS \formComma
   &\Vb
        &= \vb +\vnor\normal
         =   \begin{bmatrix}
                 \vF \\ 0
             \end{bmatrix}
            +\frac{\partial_t h + \vF \cdot \dh}{\sqrt{|g|}}\normal \in\tangentR\formPeriod
\end{empheq}
For the material derivative of a scalar field $ \psi\in\tangentScal $ we calculate
\begin{align*}
    \dot{\psi}
        &= \partial_t \psi + \nabla_{\vb-\vbobs}\psi
         = \partial_t \psi + \left( \vF - (\partial_t h)\dh \right)\cdot(\nabla\psi)^{\sharp}
         = \partial_t \psi + \vF\cdot\partial\psi
           -\frac{\partial \psi \cdot \dh}{|g|}\left( \partial_t h + \vF\cdot\dh \right)\formPeriod
\end{align*}
One may observe that $\dot{h} = \frac{\partial_t h + \vF\cdot\dh}{|g|}$ is valid 
and thus also $\dot{\psi} = \partial_t \psi + \left(\vF-\dot{h}\dh\right)\cdot\partial\psi$.
For the componentwise trace-divergence of the material velocity we get
\begin{align*}
    \DivC\Vb
        &= \partial\cdot\vF - \frac{\dh \cdot \partial \vF \cdot \dh}{|g|} 
                    - (\partial_t h + \vF \cdot \dh) \Delta h \formPeriod
\end{align*}
Eventually, the scalar Truesdell time derivative \eqref{eq:truesdell_rate} yields
\begin{empheq}[box=\fbox]{align*}
    \mathring{\psi}
        &=  \partial_t \psi - \left( \psi\Delta h+\frac{\partial\psi\cdot\dh}{|g|} \right)\partial_t h
            +\psi\left( \partial\cdot\vF - \frac{\dh\cdot\partial\vF\cdot\dh}{|g|}\right)
            +\vF\cdot\left( \partial \psi - \left( \psi\Delta h+\frac{\partial\psi\cdot\dh}{|g|} \right)\dh \right)\formPeriod
\end{empheq}
As the tangential material velocity fulfills $ \Mspat\vb = -\psi\nabla\mu $,
\ie\ with \eqref{eq:truesdell_velosubs}, we obtain
\begin{empheq}[box=\fbox]{align*}
    \mathring{\psi}
        &=  \partial_t \psi - \left( \psi\Delta h+\frac{\partial\psi\cdot\dh}{|g|} \right)\partial_t h
            -\Mspat^{-1}\psi\left( \psi\Delta\mu 
                                        + 2\left( \partial\psi \cdot \partial\mu 
                                                - \frac{(\partial\psi\cdot\dh)(\partial\mu\cdot\dh)}{|g|} \right) \right)\formPeriod
\end{empheq}

These allow to formulate the surface tension flows in height observer formulation \eqref{eq:gradient_flow_SO_height} and \eqref{eq:gradient_flow_SO_new_height}, which are appropriate for standard numerical approaches.

\section{($L^2,\HmOne$) - Surface Gradient Flows only in Normal Direction}\label{sec:gradientflow_normal}

In analogy to \cref{sec:SGF} we here consider the situation of surface evolution only in normal direction. 
We thus stipulate the ($L^2,\HmOne$) - surface gradient flow
\begin{align}\label{eq:gradient_flow_normal}
    \Mspat \vnor
            &= \normal\DX\potenergy\formComma
   &\Mdens\mathring{\psi}\vert_{\vb=\nullb}
           &= \Delta\Dpsi\potenergy\formComma
\end{align}
where
\begin{align}\label{eq:truesdell_rate_normal}
   \mathring{\psi}\vert_{\vb=\nullb}
        &= \partial_t\psi - \left(\nabla_{\vbobs}\psi + \psi\vnor\meanc\right)
\end{align}
is the scalar Truesdell time derivative \eqref{eq:truesdell_rate} for material motions solely in normal direction,
and $ \vbobs\in\tangentS $ is an arbitrary tangential observer velocity.
\begin{proposition}\label{prop:normal_fullflow_relation}
    The surface gradient flow \eqref{eq:gradient_flow_normal} equals the surface gradient flow \eqref{eq:gradient_flow} under the constraint $ \vb=\nullb $.
\end{proposition}
\begin{proof}
    The gradient flow \eqref{eq:gradient_flow_normal} is trivially obtained by substituting $\vb = \nullb$ into \eqref{eq:gradient_flow} and omitting the tangential equation. 
    However, it remains to be shown that this constraint does not induce any relevant constraint forces.
    To enforce the constraint in the gradient flow \eqref{eq:gradient_flow} in a systematic manner,
    we introduce a tangential Lagrange multiplier $\lambdab\in\tangentS$ into the spatial equation, yielding
    \begin{align*}
        \Mspat\Vb 
            &= -\DX\potenergy + \lambdab\formPeriod
    \end{align*}
    Since the constraint is independent of $\psi$, it appears only in this equation.
    Accordingly, the only constraint force is given by $\lambdab = \IdS\DX\potenergy$. 
    Since this force acts purely in the tangential direction and, under $\vb = \nullb$, completely determines the tangential equation, it is justified to omit the tangential equation, and thus the associated constraint forces.
\end{proof}
\begin{remark}
    The normal part $ \normal\DX\potenergy $ of the spatial functional derivative in \eqref{eq:gradient_flow_normal} equals the normal functional derivative under the normal Truesdell gauge of surface independence 
    \begin{align}\label{eq:truesdell_gauge_normal}
        \forall\wnor\in\tangentScal:\qquad
        \eth_{\wnor}\psi
            &:=\eth_{\wnor\normal}\psi
             = \psi\wnor\meanc\formPeriod
    \end{align}
    Consequently, the spatial equation of the gradient flow \eqref{eq:gradient_flow_normal} also follows from a purely normal (shape) variation.
\end{remark}
\begin{proof}
    Following \cite{NitschkeReutherEtAl_PRSA_2020}, the fundamental relation for local normal variations is given by
    $ \eth_{\wnor}\para = \wnor\normal $.
    A comparison with our full local variation $ \eth_{\Wb}\para = \Wb $ yields $ \wb = \nullb $.
    Moreover, it holds $ \eth_{\Wb}\psi\vert_{\wb=\nullb} = \psi\wnor\meanc = \eth_{\wnor}\psi $ for the Truesdell gauge of surface independence.
    Eventually, 
    \begin{align*}
        \innerH{\tangentScal}{\deltafrac{\potenergy}{\xi},\wnor}
            &=  \innerH{\tangentR}{\deltafrac{\potenergy}{\para}, \wnor\normal}
             = \innerH{\tangentScal}{\normal\DX\potenergy,\wnor}\formComma
    \end{align*}
    where $ \xi $ is the deformation coordinate in normal direction, gives the assumption.
\end{proof}
\begin{proposition}\label{prop:props_of_gradient_flow_normal}
    The $ (L^2, \HmOne) $- surface gradient flow \eqref{eq:gradient_flow_normal}, \wrt\ the normal Truesdell gauge of surface independence \eqref{eq:truesdell_gauge_normal}, 
    ensures the conservation property 
    \begin{align*}
        \ddt \int_{\surf}\psi\dS
             &= 0 \formComma
    \end{align*}
    and a proper energy dissipation
    \begin{align*}
        \ddt\potenergy
            &= -\left( \Mspat^{-1} \normHsq{\tangentScal}{\normal\DX\potenergy} 
                      +\Mdens^{-1} \normHsq{\tangentS}{\nabla\Dpsi\potenergy}\right) 
             \le 0\formPeriod
    \end{align*}
\end{proposition}
\begin{proof}
    \cref{prop:normal_fullflow_relation} and \eqref{eq:ddtU_1} yields
    \begin{align*}
        \ddt\potenergy
            &= \innerH{\tangentR}{\DX\potenergy, \Vb}
                    + \innerH{\tangentScal}{\Dpsi\potenergy, \mathring{\psi}}
             = \innerH{\tangentR}{\normal\DX\potenergy, \vnor}
                     + \innerH{\tangentScal}{\Dpsi\potenergy, \mathring{\psi}}\formPeriod
    \end{align*}
    Substituting the gradient flow  \eqref{eq:gradient_flow_normal} into this expression completes the proof.
\end{proof}
\begin{example}\label{ex:surface_tension_flow_normal}
    The surface energy $ \energyHT = \int_{\surf} f(\psi) \dS $ \eqref{eq:energyHT} yields the normal surface tension flow
    \begin{align*}
        \Mspat \vnor
            &= \left( f(\psi) - \psi f'(\psi)\right)\meanc\formComma\\
        \Mdens\left( \partial_t\psi - \left( \nabla_{\vbobs}\psi + \vnor\psi\meanc \right) \right)
                &= \div(f''(\psi)\nabla\psi)
                 =f''(\psi)\Delta\psi + f'''(\psi)\normsq{}{\nabla\psi}\formComma
    \end{align*}
    due to \cref{prop:normal_fullflow_relation}, \eqref{eq:S_gradientflow}, and \eqref{eq:truesdell_rate_normal}.
    Equivalently, using the surface tension $\sigma(\psi) =   f(\psi) - \psi f'(\psi)$ \eqref{eq:surface_tension}, this results in
    \begin{align*}
        \Mspat \vnor
             &= \sigma(\psi)\meanc\formComma\\
        \Mdens\left( \partial_t\psi - \left( \nabla_{\vbobs}\psi + \vnor\psi\meanc \right) \right)
            &= -\div\left(\frac{\sigma'(\psi)}{\psi}\nabla\psi\right)
             = -\frac{\sigma'(\psi)}{\psi}\Delta\psi
                +\frac{\sigma'(\psi)-\psi\sigma''(\psi)}{\psi^2}\normsq{}{\nabla\psi}\formPeriod
    \end{align*}
    The height-observer formulation the corresponding gradient flow to \eqref{eq:gradient_flow_SO_height} becomes
    \begin{align}\label{eq:gradient_flow_SO_height_normal}
    \begin{aligned} 
            \Mspat \partial_t h
                &= |g|\left(  f(\psi) - \psi f'(\psi)\right)\hfrak\formComma\\
            \Mdens\left(\partial_t \psi - \left( \psi\hfrak+\frac{\partial\psi\cdot\dh}{|g|} \right)\partial_t h\right)\hspace{-5em}\\
                &=  f''(\psi)
                      \left(\partial\cdot\partial\psi - (\partial\psi \cdot\dh) \hfrak - \frac{\dh\cdot\partial^2\psi\cdot\dh}{|g|} \right)
                    +  f'''(\psi)\left( \partial\psi\cdot\partial\psi - \frac{(\partial\psi\cdot\dh)^2}{|g|} \right)\formComma
        \end{aligned}
    \end{align}
    where $ \hfrak = \frac{\partial\cdot\dh}{|g|} - \frac{\dh\cdot\ddh\cdot\dh}{|g|^2} $, and the normal velocity is given by $ \vnor = \frac{\partial_t h}{\sqrt{|g|}} $.
\end{example}

\noindent
{\bfseries Acknowledgements}: A.V. acknowledges support by the German Research Foundation (DFG) through FOR3013 "Vector- and tensor-values surface PDEs".

%\clearpage
\printbibliography

\end{document}